\documentclass[11pt]{article}
\usepackage{fullpage}
\usepackage{epsfig}
\usepackage{graphics}
\usepackage{latexsym}
\usepackage{amsmath}
\usepackage{amsfonts}
\usepackage{amssymb}
\usepackage{mathrsfs}
\usepackage{amsthm}
\usepackage{xspace}
\usepackage{epstopdf}
\usepackage{float}
\usepackage{hyperref}

\numberwithin{equation}{section}
\usepackage[ruled,vlined]{algorithm2e}
\SetArgSty{textrm}

\usepackage[english]{babel}
\usepackage[nottoc]{tocbibind}

\usepackage{caption}
\usepackage{subcaption}

\usepackage[usenames,dvipsnames]{color}
\usepackage{pifont}
\usepackage{booktabs}
\usepackage{multirow}

\newcommand{\redcomment}[1]{\textcolor{red}{\textrm{#1}}}

\usepackage{amsthm}     
\theoremstyle{plain}     
\newtheorem{theorem}{Theorem}

\newtheorem{lemma}{Lemma}
\newtheorem{proposition}{Proposition}
\theoremstyle{definition} 

\newtheorem{mechanism}{Mechanism}
\theoremstyle{remark} 

\makeatletter
\@addtoreset{equation}{section}
\def\section{\@startsection {section}{1}{\z@}{-3.5ex plus -1ex minus
 -.2ex}{2.3ex plus .2ex}{\large\bf}}
\makeatother

\def\bfm#1{\mbox{\boldmath$#1$}}

\def\0{\bfm 0}

\DeclareMathAlphabet{\mathpzc}{OT1}{pzc}{m}{it}

\newcounter{my}

\newcounter{my2}

\newcounter{my3}

\newcounter{my4}

\newcounter{my5}

\newcounter{my6}

\allowdisplaybreaks 

\begin{document}

\title{Strategyproof Facility Location with Prediction: Minimizing the Maximum Cost}


\date{}
\maketitle

\vspace{-3em}
\begin{center}

\author{Hau Chan$^{1}$\quad Jianan Lin$^{2}$\quad Chenhao Wang $^{3,4}$\\
${}$\\
$1$ University of Nebraska-Lincoln\\
$2$ Rensselaer Polytechnic Institute\\
$3$ Beijing Normal University-Zhuhai\\
$4$ Beijing Normal-Hong Kong Baptist University\\
\medskip
}

\end{center}

\begin{abstract}
We study the mechanism design problem of facility location on a metric space in the learning-augmented framework, where mechanisms have access to imperfect predictions of the optimal facility locations. Our objective is to design strategyproof (SP) mechanisms that truthfully elicit agents’ preferences over facility locations and, using the given prediction, select a facility location that approximately minimizes the maximum cost among all agents. In particular, we seek SP mechanisms whose approximation guarantees depend on the prediction error: they should achieve improved performance when the prediction is accurate (the property of \emph{consistency}) while still ensuring strong worst-case guarantees when the prediction is arbitrarily inaccurate (the property of \emph{robustness}).

On the real line, we characterize all deterministic SP mechanisms with consistency strictly better than 2 and bounded robustness for the maximum cost. We show that any such mechanism must coincide with the MinMaxP mechanism, which returns the prediction if it lies between the two extreme agent locations and otherwise returns the agent location closest to the prediction.
For any prediction error $\eta\ge 0$, we prove that MinMaxP achieves a $(1+\min(1, \eta))$-approximation and that no deterministic SP mechanism can obtain a better approximation ratio.
In addition, for two-dimensional spaces with the $\ell_p$ distance, we analyze the approximation guarantees of a deterministic mechanism that applies MinMaxP independently on each coordinate, as well as a randomized mechanism that selects between two deterministic mechanisms with carefully chosen probabilities.
We further extend these results to the $L_p$-norm social cost objective on the line metric and the maximum cost objective on the tree metric. 
Finally, we examine the group strategyproofness of the mechanisms.
\end{abstract}

\section{Introduction}\label{sec:intro}

The facility location problem is a central topic in operations research, economics, and computer science, with applications ranging from urban planning to server deployment \cite{hochbaum1982heuristics,DBLP:journals/scw/Miyagawa01,Frank2007}. 
In this problem, a typical aim is to determine the placement of one or more facilities to best serve a set of agents, based on the agents’ location preferences within a given metric space (e.g., a street or an area).

Over the past decades, the facility location problem has attracted extensive interest at the intersection of artificial intelligence and mechanism design \cite{chan2021mechanismsurvey}, primarily focusing on agents’ strategic behavior, where agents may benefit from misreporting their location preferences to manipulate facility locations. The main goal and challenge in this strategic version of the problem is therefore to design strategyproof (SP) mechanisms that truthfully elicit agents’ location preferences—so that agents have no incentive to misreport their locations—while determining facility locations that approximately optimize certain objectives. Pioneering works include that of Moulin \cite{moulin1980strategy} and Procaccia and Tennenholtz \cite{procaccia2013approximate}.

This strategic version of the facility location problem is the first well-studied case of the paradigm of approximate mechanism design without money \cite{procaccia2013approximate}. Since these initial studies, a significant body of research has focused on designing SP mechanisms to determine facility locations that (approximately) optimize the social cost and maximum cost objectives, which measure the total cost and the maximum cost over agents, respectively. However, due to standard theoretical assumptions and limitations, most classical work evaluates mechanisms by their worst-case performance, typically through the approximation ratio. This focus often leads to impossibility results that do not fully capture typical behavior in practice, where additional structure or partial information is often available \cite{roughgarden2021beyond}.

To overcome these practical limitations, a powerful framework of machine learning (ML) augmented algorithms has been introduced in broader algorithm design \cite{lykouris2021competitive}. This framework aims to surpass worst-case bounds by incorporating predictions (e.g., about optimal solutions or future inputs) as supplementary advice. The goal is to design algorithms that perform exceptionally well when predictions are accurate (achieving consistency), yet degrade gracefully and maintain a robust baseline performance even when predictions are poor (achieving robustness). This paradigm has proven successful in areas like online and streaming algorithms \cite{mitzenmacher2022algorithms}, offering a principled way to leverage often-available side information—such as historical or contextual data—to break classical impossibility results.

Recently, Agrawal et al. \cite{agrawal2022learning} and Xu and Lu \cite{DBLP:conf/ijcai/XuL22} proposed the study of \emph{learning-augmented} mechanism design to address these limitations of classical results. In the general learning-augmented framework, mechanism designers are provided with a prediction (e.g., produced by a machine learning method) about some aspect of an optimal solution (see the survey by Mitzenmacher and Vassilvitskii \cite{mitzenmacher2022algorithms}). For instance, in facility location, such predictions could estimate optimal facility locations or agent locations based on historical or environmental data. The works \cite{agrawal2022learning,DBLP:conf/ijcai/XuL22} leverage predictions to design SP mechanisms and demonstrate that SP mechanisms with access to predictions can achieve strictly better approximation guarantees when the prediction is accurate (i.e., \emph{consistency}), while matching classical approximation guarantees even when the prediction is arbitrarily inaccurate (i.e., \emph{robustness}).

\subsection{Our Contributions}

In this paper, we study learning-augmented mechanism design for facility location on the one-dimensional real line \(\mathbb{R}\), on trees, and in two-dimensional spaces equipped with the \(\ell_p\) metric (denoted by \(\ell_2^p\) space). Our goal is to leverage (possibly imperfect) predictions of optimal facility locations to design both deterministic and randomized SP mechanisms under the \emph{maximum cost objective}. Unless otherwise stated, all mechanisms discussed in this paper are assumed to be \emph{anonymous} (their outputs do not depend on the identities of the agents) and \emph{unanimous} (if all agents report the same location, the mechanism outputs that location). It is clear that any mechanism with bounded robustness must be unanimous, because when all the agents are in the same location, the optimal maximum cost is 0 and a mechanism with bounded robustness has to output this location. Note that unanimity excludes the mechanism that always returns the prediction location. Below, we summarize our main contributions.

\begin{itemize}

    \item For the metric space of the real line $\mathbb R$, we obtain the following results:

    \begin{itemize}

        \item We analyze the MinMaxP mechanism of Agrawal et al. \cite{agrawal2022learning}, which returns the prediction location if it lies between the two extreme agent locations, and otherwise returns the closest agent location to the prediction. For  the maximum cost
objective,  while MinMaxP achieves a $(1+\min(1, \eta))$-approximation as a function of the prediction error $\eta$ \cite{agrawal2022learning}, we show that this bound is tight: no deterministic SP mechanism can outperform this bound for any $\eta\ge 0$;  

        \item We completely characterize deterministic SP mechanisms with acceptable performance guarantees: for any $\epsilon>0$, MinMaxP is the unique deterministic SP mechanism achieving both $(2-\epsilon)$-consistency and bounded robustness;  

        \item For randomized SP mechanisms, we study a mechanism, proposed by Balkanski et al. \cite{balkanski2024randomized}, that randomly selects MinMaxP (with probability $1-q$) and a classical prediction-free SP mechanism (with probability $q$). We prove that it achieves an approximation ratio of $1+\frac{q}{2}+(1-q)\min(1,\eta)$, and additionally, we establish a lower bound showing that no randomized SP mechanism can beat a $\frac32$-approximation when $\eta\ge \frac{1}{2}$.

         \item In addition, we consider the $L_p$-norm social cost objective. We show the MinMaxP mechanism is 1-consistent and $\mathcal{O}(n^{1/p})$-robust for the $L_p$-norm social cost objective. When $1<p<+\infty$, MinMaxP is the unique deterministic SP mechanism with $1$-consistency and bounded robustness.
    \end{itemize}

       \item For tree metrics, where the distance between two points \(x, y\) is defined as the length of the unique path connecting them, we obtain the following results for the maximum cost:

    \begin{itemize}
        \item For deterministic SP mechanisms, we show that a variant of  MinMaxP achieves a tight $(1+\min(1,\eta))$-approximation, which implies 1-consistency and 2-robustness. 
    
        \item For randomized SP mechanisms, we show that a mixed mechanism—which runs MinMaxP with probability $1-q$ and the Tree Random mechanism of \cite{alon2010strategyproof} with probability $q$—achieves an approximation ratio of 
        $1+\frac{n}{n+2}q + (1-q)\min(1, \eta)$.
    \end{itemize}

    \item For $\ell_2^p$ spaces ($p\ge 2$), where the distance between $x,y\in\mathbb R^2$ is defined by $d_p(x, y) =(\sum_{i=1}^{2} |x_i - y_i|^p)^{1/p}$, we provide the following new bounds:

    \begin{itemize}
        \item For deterministic SP mechanisms, we analyze the Minimum Bounding Box mechanism, which runs MinMaxP independently on each coordinate, and we show that it has a $(1+\min(2^{1/p}, \eta))$-approximation.  On the negative side, while no deterministic SP mechanism is $(2-\epsilon)$-consistent and $(1+2^{1/p}-\epsilon)$-robust when $p=2$ \cite{agrawal2022learning}, extending this impossibility result to general $p\ge 2$ remains challenging. We conjecture, however, that the same lower bound applies for all $\ell_2^p$ spaces with $p\ge 2$.
    
        \item  For randomized SP mechanisms, we show that a mixed SP mechanism, which runs Minimum Bounding Box with probability $1-q$ and the coordinate-wise median with probability $q$, achieves an approximation ratio of $1+q + (1-q)\min(2^{1/p}, \eta)$. 
    \end{itemize}


\end{itemize}

We further investigate the group strategyproofness (GSP) and strong group strategyproofness (SGSP) of the mechanisms we study. On the real line and on trees, we show that the MinMaxP mechanism is SGSP, and that no SGSP mechanism can achieve a robustness guarantee better than 2. Thus, MinMaxP attains the optimal robustness guarantee among all SGSP mechanisms in these settings. In contrast, in \(\ell_2^p\) spaces, none of the mechanisms we consider are GSP, and for any prediction error \(\eta \ge 0\), no (possibly randomized) GSP mechanism can achieve an approximation ratio better than 2.

\paragraph{Organization.} In Section \ref{sec:ore}, we present the preliminaries of the considered mechanism design setting. 
In Sections \ref{sec:line}, \ref{sec:tree} and \ref{sec:2d}, we study the line, tree and $\ell_2^p$ spaces, respectively. 
In Section \ref{sec:ext} we discuss the group strategyproofness of mechanisms.

Compared to the preliminary conference version \cite{ecai2025} of this work, this full version additionally considers \(L_p\)-norm social cost objectives for the line and maximum cost objectives for tree metrics, which introduce further challenges beyond the classical objectives on the real line.

\subsection{Related Work}


\paragraph{Mechanism Design for Facility Location.}
The study of mechanism design for facility location has a rich and extensive history. Moulin \cite{moulin1980strategy} provided a seminal characterization of deterministic SP mechanisms for facility location on the real line. Such characterization includes the median mechanism, which places the facility at the median of the elicited agent locations. 
Moulin \cite{moulin1980strategy} further showed that all deterministic SP mechanisms on the line are general median mechanisms. 
Procaccia and Tennenholtz \cite{procaccia2013approximate} introduced the paradigm of approximate mechanism design without money, highlighting the trade-off between strategyproofness and approximation guarantees. 
For the social cost, the median mechanism is SP and optimal. 
They proved that the median mechanism provides a tight $2$-approximation for the maximum cost, and that a randomized mechanism, LRM, achieves a $1.5$-approximation, alongside tight lower bounds.  
In the two-dimensional Euclidean spaces, the coordinate-wise median (CM) mechanism provides $\sqrt2$-approximation for the social cost \cite{DBLP:conf/sagt/Meir19} and 2-approximation for the maximum cost  \cite{DBLP:journals/scw/GoelH23}, both of which are optimal among deterministic SP mechanisms. 
Other notable extensions include obnoxious facility location \cite{cheng2013strategy,Ibara2012characterize,ye2015strategy}, other metric spaces \cite{alon2010strategyproof,feldman2013strategyproof,lin2020nearly,gravin2025approximation,li2017mechanism}, settings with facility preferences \cite{fong2018facility,PaoloCarmine2016}, distance constraints \cite{chen2018mechanism,chen2021tight}, and alternative cost functions \cite{feldman2013strategyproof,fotakis2013strategyproof}. 
For a comprehensive survey, see Chan et al. \cite{chan2021mechanismsurvey}.

\paragraph{Learning-augmented Mechanism Design.} 
The learning-augmented, or algorithms-with-predictions paradigm \cite{mitzenmacher2022algorithms}, enriches traditional results by integrating machine-learned predictions into the algorithmic design. 
Most research in this area has concentrated on online optimization, but learning-augmented mechanism design has recently emerged through the initial works of Agrawal et al. \cite{agrawal2022learning} and Xu and Lu  \cite{DBLP:conf/ijcai/XuL22}. 
Focusing on mechanism design for facility location, Agrawal et al. introduced mechanisms that leverage predictions to balance between consistency  and robustness, providing approximation guarantees as a function of prediction error. 

Subsequent works have considered other settings of learning-augmented mechanism design for facility location.  Barak et al. \cite{DBLP:conf/nips/Barak0T24} considered prediction models where most agents' locations are well-predicted but with a few arbitrary outliers. 
Chen et al. \cite{DBLP:journals/corr/abs-2410-07497} studied the price of anarchy in the game induced by a mechanism called \textit{Harmonic} in general metric spaces. \cite{nips2025envy} discussed the learning-augmented mechanisms for envy ratio.
The variant of obnoxious facility location—where agents prefer the facility as far away as possible—has been investigated by Istrate and Bonchis \cite{DBLP:journals/corr/abs-2212-09521} and Fang et al. \cite{DBLP:conf/tamc/FangFLN24} under learning-augmented mechanism design. 
Christodoulou et al. \cite{DBLP:conf/nips/000SV24} introduced the notion of \emph{quality of recommendation} as an alternative performance metric. Balkanski et al. \cite{balkanski2024randomized} explored randomized mechanisms and examined how different types of predictions impact robustness and consistency, particularly for the maximum cost.

\medskip
Compared to existing work on learning-augmented mechanism design for facility location, to the best of our knowledge, our paper is the first to provide a complete characterization of SP mechanisms with acceptable performance guarantees, and the first to consider \(\ell_2^p\) spaces. Moreover, whereas most prior studies focus only on individual strategyproofness, we offer a deeper analysis of (strong) group strategyproofness.

\section{Preliminaries}\label{sec:ore}

In the considered mechanism design setting, we have a set $N=\{1,2,\ldots,n\}$ of $n$ agents on a metric space $(\mathcal M,d)$ with agent location (preference) profile $\mathbf{x} = (x_1, x_2, \ldots, x_n)\in\mathcal M^n$, where $d$ is the distance function such that the distance between two points $x, y\in\mathcal M$ is  $d(x, y)$. 
Given the facility location $y\in\mathcal M$, the cost of each agent $i\in N$ is the distance from their location to the facility, i.e., $c(x_i, y) = d(x_i, y)$. 
We want to minimize the maximum cost of all agents:
\begin{align*}
    \text{MC}(\mathbf{x}, y) &= \max_{i\in N}c(x_i, y).
\end{align*}

In the setting of  mechanism design with predictions, we are given a prediction $\pi\in\mathcal M$ on the optimal facility location. 
A deterministic mechanism is a function $f$ : $\mathcal M^n \times \mathcal M\rightarrow \mathcal M$ that maps the location profile $\mathbf{x}$ of agents and the prediction $\pi$ to a facility location $f(\mathbf{x}, \pi) = y$. 
A randomized mechanism is a function $f$ : $\mathcal M^n \times \mathcal M\rightarrow \Delta(\mathcal M)$ that maps the location profile and the prediction to a probability distribution, where $\Delta(\mathcal M)$ denotes the set of all distributions over $\mathcal M$. 
Given a distribution $f(\mathbf{x}, \pi)\in \Delta(\mathcal M)$, the (expected) cost of agent $i$ is $c(x_i, y)=\mathbb{E}_{y\sim f(\mathbf{x, \pi})}[d(x_i, y)]$.

A mechanism is \emph{strategyproof} (SP) if no single agent can
benefit from misreporting their location, regardless of the reported locations of the other agents.
Formally, a mechanism $f$ is SP if for any location profile $\mathbf x$, prediction $\pi$, agent $i\in N$ and $x_i'\in\mathcal M$, it holds that $c(x_i, f(x_i', \mathbf{x}_{-i}, \pi)) \ge c(x_i, f(\mathbf{x}, \pi))$. 
A mechanism is \textit{unanimous} if when all agents have the same location, it outputs this location. Formally, a unanimous mechanism guarantees that $f(\mathbf{x}, \pi)=c$ whenever $\forall i\in N, x_i=c$. 

For any location profile $\mathbf x$, let $o(\mathbf x)$ be an optimal solution that minimizes the maximum cost. 
We measure the performance or guarantee of a mechanism by its consistency and robustness.
A mechanism $f$ is $\alpha$-\textit{consistent} if it has an $\alpha$-approximation when the prediction is correct (i.e., $\pi=o(\mathbf x)$):
$$
\max_{\mathbf x}\frac{\text{MC}(\mathbf x, f(\mathbf x, o(\mathbf x)))}{\text{MC}(\mathbf x, o(\mathbf x))} \le \alpha.
$$
On the other hand, $f$ is $\beta$-\textit{robust} if it has a $\beta$-approximation when the predictions can be arbitrarily wrong:
$$
\max_{\mathbf x, \pi}\frac{\text{MC}(\mathbf x, f(\mathbf x, \pi))}{\text{MC}(\mathbf x, o(\mathbf x))} \le \beta.
$$

Furthermore, there is a more general approximation guarantee for mechanisms
as a function of the \emph{prediction error}.
Following the standard definition \cite{agrawal2022learning}, let the error be the distance between the prediction location and an optimal location, divided by the optimal maximum cost:
$$
\eta(\pi, \mathbf{x}) = \frac{d(o(\mathbf{x}),\pi)}{\text{MC}(\mathbf{x},o(\mathbf{x}))}.
$$
Given a bound $\eta$ on the prediction error, a mechanism achieves a $\gamma(\eta)$-approximation if
$$
\max_{\pi, \mathbf{x}: \eta(\pi, \mathbf{x})\le \eta}\frac{\text{MC}(\mathbf x, f(\mathbf x, \pi))}{\text{MC}(\mathbf x, \pi)} \le \gamma(\eta).
$$
When $\eta=0$, this worst-case guarantee is the consistency. 
The worst-case guarantee, irrespective of the error (i.e., over all values of $\eta$), is the robustness.

\section{The Real Line}\label{sec:line}
In this section, we consider the mechanism design problem on the real line $\mathbb R$, i.e., the space is $\mathcal M=\mathbb R$.  
For convenience, assume that $x_1\le \cdots\le x_n$. 
In Section \ref{sec:det} and Section \ref{sec:rand}, we study deterministic SP mechanisms and randomized SP mechanisms, respectively. 

\subsection{Deterministic Mechanisms}\label{sec:det}

We consider the following SP mechanism (called MinMaxP) proposed by Agrawal et al. \cite{agrawal2022learning}. 
They show that MinMaxP has $(1+\min(1, \eta))$-approximation with respect to prediction error $\eta$ for the maximum cost, and  thus it is $1$-consistent and $2$-robust.

\begin{mechanism}[MinMaxP \cite{agrawal2022learning}]\label{mec:line}
 Given location profile $\mathbf{x}$ and prediction $\pi$, return 
 $$y = \max(x_1, \min(x_n, \pi)).$$
\end{mechanism}

\paragraph{Tight lower bound.} It is well known that in the prediction-free setting \cite{procaccia2013approximate}, no deterministic SP mechanism has an approximation ratio less than 2. Because the robustness is defined by how the mechanism performs when predictions are arbitrarily bad, it indicates that the robustness of any SP mechanism with prediction cannot be better than $2$. Since 1-consistency is trivially optimal,  MinMaxP is the best possible for both consistency and robustness. 

While the work of \cite{agrawal2022learning} does not provide a lower bound on the approximation ratios with respect to the error bound $\eta$, we complete their results by giving a matching lower bound. 
That is, we show that the $(1+\min(1, \eta))$-approximation of MinMaxP is best possible for any $\eta\ge 0$. 

Before proceeding, we present several properties of SP mechanisms that hold regardless of whether a prediction is present. The first property, due to Border and Jordan \cite{border1983}, states that any SP mechanism must be "uncompromising".  Formally, a mechanism $f$ is said to be \emph{uncompromising} if for every location profile $\mathbf x$ and each agent $i\in N$, if $f(\mathbf x)=y$, then 
\begin{align*}
x_i > y &\implies f(x_i', \mathbf{x}_{-i}) = y &&\text{for all } x_i' \geq y, ~\text{and}~ \\
x_i < y &\implies f(x_i', \mathbf{x}_{-i}) = y &&\text{for all } x_i' \leq y.
\end{align*}

\begin{lemma}[\cite{border1983}]\label{lemma:sp}
Any deterministic SP mechanism  must be uncompromising. 
\end{lemma}

\begin{lemma}\label{lemma:line-5}
  For any location profile $\mathbf x$,  any deterministic SP mechanism  always outputs $y\in [x_1, x_n]$.
\end{lemma}
\begin{proof}
    Suppose there exists a mechanism that outputs $y<x_1$ or $y>x_n$ for some instance. W.l.o.g. assume that $y<x_1$. Then we move all the agents to $x_1$ one by one, and according to Lemma \ref{lemma:sp}, the output $y$ does not change. This contradicts the unanimity.
\end{proof}

Now we are ready to present the lower bound result on all deterministic SP mechanisms. 

\begin{theorem}\label{thm:lowl}
    For any prediction error bound $\eta\ge 0$, no deterministic SP mechanism can achieve an approximation ratio better than $1+\min(1, \eta)$.
\end{theorem}

\begin{proof}
   Suppose $f$ is such a mechanism with approximation ratio $1+\min(1, \eta)-\delta$ with respect to prediction error bound $\eta\ge 0$, where $\delta>0$ is a constant. We discuss three cases, where the case when $\eta=0$ is trivial. 

       \textbf{Case 1}: $\eta \in(0, 1)$. By the definition of $\eta$, the prediction $\pi$ lies in interval $(x_1,x_n)$.  Assume without loss of generality that $\frac{x_1+x_n}{2}\le \pi < x_n$. By the approximation ratio $f$ and the definition of $\eta$ we know that $y\in (x_1+x_n-\pi, \pi)$. Then we move $x_n$ to $x_n'=x_n+C_1$ where $C_1\rightarrow\infty$ is large enough. According to Lemma \ref{lemma:sp}, we know $y$ does not change. In the same way, we move $x_1$ to $x_1'=x_1-C_2$ where $C_2=C_1\cdot \frac{\pi-x_1}{x_n-\pi}>C_1$ is also large enough, and $y$ still does not change. In the new instance, the prediction error remains the same because
    \begin{align*}
        \frac{|\pi-\frac{x_1'+x_n'}{2}|}{x_n'-x_1'} = \frac{\pi-\frac{x_1+x_n}{2}+\frac{C_2-C_1}{2}}{x_n-x_1+C_1+C_2}=\frac{\frac{C_2-C_1}{2}}{C_1+C_2}=\frac{\pi-\frac{x_1+x_n}{2}}{x_n-x_1},
    \end{align*}
    where the second equation is due to
    \begin{align*}
        \frac{C_2-C_1}{C_2+C_1}=\frac{\frac{C_2}{C_1}-1}{\frac{C_2}{C_1}+1}=\frac{\frac{\pi-x_1}{x_n-\pi}-1}{\frac{\pi-x_1}{x_n-\pi}+1}=\frac{2\pi-(x_1+x_n)}{x_n-x_1}.
    \end{align*}
    The approximation ratio on the new instance is 
    \begin{align*}
        \frac{MC(\mathbf x',y)}{o(\mathbf x')}&\ge \frac{y-x_1'}{\frac{x_n'-x_1'}{2}}> \frac{x_1+x_n-\pi-x_1'}{\frac{x_n'-x_1'}{2}} = \frac{2(x_n-\pi+C_2)}{x_n-x_1+C_1+C_2},
    \end{align*}
    which approaches $\frac{2C_2}{C_1+C_2}=1+\eta$ when $C_1,C_2$ go to infinity and the prediction $\pi$ has error $\eta$. It derives a contradiction to the approximation ratio $1+\eta-\delta$.

    \textbf{Case 2}: $\eta = 1$. Assume w.l.o.g. that the prediction is $\pi=x_n$, with an error equal to 1. To have an approximation ratio less than 2, it must be $y\in (x_1, x_n)$. We move all agents with $x_i < y$  to $y$ one by one. According to Lemma \ref{lemma:sp}, $y$ does not change, and the error retains 1 in the new instance. It is easy to see that the ratio of the new instance is 2, a contradiction to the approximation ratio $2-\delta$. 

    \textbf{Case 3}: $\eta > 1$.  Assume w.l.o.g. that the prediction is $\pi>x_n$ with an error $\eta$. Obviously $y\in (x_1, x_n)$. We move all agents with $x_i < y$  to $y$ one by one. According to Lemma \ref{lemma:sp}, $y$ does not change. We know $x_1'=y$. Then we move $x_n$ to $x_n'$ so that $\frac{\pi-x_n'}{\pi-x_1'}=\frac{\pi-x_n}{\pi-x_1}$. It is easy to see that $x_n<x_n'<\pi$, and the error $\eta$ does not change in the new instance. By Lemma \ref{lemma:sp}, $y$ does not change.  However, the ratio of the new instance is 2, a contradiction to the approximation ratio $2-\delta$. 
\end{proof}

\paragraph{Characterizations.} Next, we provide a characterization of all deterministic SP mechanisms with acceptable performance guarantees, that is, they have to be the MinMaxP mechanism. The following three lemmas says that such a mechanism must output the prediction as the facility location, whenever the prediction is correct, is between the two extreme agent locations, or is one of the extreme agent locations, respectively. 

\begin{lemma}\label{lemma:line-1}
    If the prediction is correct (i.e., $\pi=\frac{x_1 + x_n}{2}$), then for any $\epsilon\in (0, 1]$, any deterministic SP $(2-\epsilon)$-consistent mechanism must output $y=\pi$.
\end{lemma}

\begin{proof}
    Suppose  $y> \pi$ for some instance with correct prediction $\pi=\frac{x_1 + x_n}{2}$. By Lemma \ref{lemma:line-5}, we know $y\in (\pi, x_n]$. Then we move those agents with location $x_i>y$ to $y$ one by one. By Lemma \ref{lemma:sp}, $y$ does not change. Next, we move those agents with location $x_i<2\pi-y$ to $2\pi-y$ one by one, and do not change the output $y$. While in this new instance $\mathbf x'$ the prediction is still correct, i.e., $\pi=\frac{x_1'+x_n'}{2}$, the ratio induced by $y$ is 2, a contradiction to the $(2-\epsilon)$-consistency. 
\end{proof}

\begin{lemma}\label{lemma:line-2}
    If $x_1 < \pi < x_n$, then for any $\epsilon\in (0, 1]$, any deterministic SP $(2-\epsilon)$-consistent mechanism must output $y=\pi$.
\end{lemma}

\begin{proof}
    It follows from Lemma \ref{lemma:line-5} that $y\in (x_1, x_n)$. If $x_1 + x_n = 2\pi$, it reduces to Lemma \ref{lemma:line-1}. We assume that $x_1 + x_n > 2\pi$.
   If the output is $y> \pi$,  agent 1 can decrease the cost by misreporting $x_1'=2\pi-x_n$ such that $\pi=\frac{x_1'+x_n}{2}$ is a correct prediction and thus the output becomes $y'=\pi$ by  Lemma \ref{lemma:line-1}, a contradiction to the strategyproofness. If $y< \pi$,  we can move those agents of locations $x_i>2\pi-x_1$ to $2\pi-x_1$ one by one, and do not change the output $y$ by Lemma \ref{lemma:sp}. In this new instance, $\pi$ is a correct prediction, and the output has to be $\pi$ by Lemma \ref{lemma:line-1}, which leads to a contradiction. Therefore, the only possibility is $y=\pi$.
\end{proof}

\begin{lemma}\label{lemma:line-3}
    If $\pi= x_1$ or $x_n$, then any deterministic  SP $(2-\epsilon)$-consistent mechanism must output $y=\pi$, where $\epsilon\in (0, 1]$.
\end{lemma}

\begin{proof}
We only consider the case when $\pi= x_1$. If the output is $y>x_1$, then agent 1 can decrease the cost to 0 by misreporting the location as $x_1'=x_1-\delta$ for some sufficiently small $\delta>0$ such that the output becomes $y'=\pi$ by Lemma \ref{lemma:line-2}, resulting in a contradiction to the strategyproofness. It follows from Lemma  \ref{lemma:line-5} that $y=x_1$. 
\end{proof}

We now show that MinMaxP is the unique mechanism with desired performance guarantees. 

\begin{theorem}\label{thm:cha}
    For any $\epsilon\in (0, 1]$, MinMaxP is the only deterministic SP mechanism with $(2-\epsilon)$-consistency and bounded robustness.
\end{theorem}

\begin{proof}
    Let $f$ be a deterministic SP  mechanism with $(2-\epsilon)$-consistency and bounded robustness (which implies the unanimity). When the prediction $\pi$ lies in the interval $[x_1,x_n]$, it follows from Lemma \ref{lemma:line-2} and \ref{lemma:line-3} that the output must be $\pi$, the same as that in MinMaxP. 

    When $\pi<x_1$, if $d(x_1, y)> x_1-\pi$, then agent 1 can decrease the cost by misreporting $x_1'=\pi$ such that the output becomes $\pi$. Hence, it must be  $d(x_1, y)\le x_1-\pi$. Since $y\in [x_1, x_n]$,  we have $y\in [x_1, \min(x_n, 2x_1-\pi)]$.  If $y>x_1$,   let agent 1 move to $x_1'=\pi$, and by Lemma \ref{lemma:sp} the output for the new instance remains to be $y$. While Lemma \ref{lemma:line-3} claims that the output should be the prediction $\pi=x_1'$, there is a contradiction.  Therefore, the mechanism must output $y=x_1$ when $\pi<x_1$, the same as that in MinMaxP. The symmetric arguments work for the case when $\pi>x_n$.
\end{proof}

\subsection{Randomized Mechanisms}\label{sec:rand}

The characterization in Theorem \ref{thm:cha} does not apply to randomized mechanisms. There is a randomized SP, 1-consistent and 2-robust mechanism: If $x_1\le \pi\le x_n$, return $y=\pi$. If $\pi<x_1$,  return $y=x_1$ with probability $\max(0.5, 1-\frac{x_1-\pi}{x_n-x_1})$ and $y=x_n$ with probability $\min(0.5, \frac{x_1-\pi}{x_n-x_1})$. If $\pi>x_n$,  return $y=x_n$ with probability $\max(0.5, 1-\frac{\pi-x_n}{x_n-x_1})$ and $y=x_1$ with probability $\min(0.5, \frac{\pi-x_n}{x_n-x_1})$. The 1-consistency and 2-robustness are clear, and the strategyproofness can also be verified. 

Consider the following classic randomized mechanism without predictions by Procaccia and  Tennenholtz \cite{procaccia2013approximate}, which is SP and $\frac{3}{2}$-approximate for the maximum cost. 

\begin{mechanism}[LRM]\label{mec:line-classic}
     Given profile $\mathbf{x}$, return $x_1, x_n, \frac{x_1+x_n}{2}$ with probability $\frac{1}{4}, \frac{1}{4}, \frac{1}{2}$ respectively.
\end{mechanism}

In the prediction setting, we first give a lower bound $\frac{3}{2}$ on the approximation ratio of randomized mechanisms when the prediction error bound is no less than $\frac12$. It indicates that the prediction-free $\frac{3}{2}$-approximation LRM is indeed the best possible when the prediction is not very accurate, and we could discard the prediction information without harming the performance guarantee.

\begin{theorem}
    When the prediction error bound is {$\eta \ge \frac{1}{2}$}, no randomized SP mechanism has an approximation ratio better than $\frac{3}{2}$.
\end{theorem}
\begin{proof}
    Consider $n=2$ agents and the location profile $\mathbf{x}=(0, 2)$. The prediction is $\pi=1$.
    The output $y$ satisfies either $\mathbb{E}[|y-x_1|]\ge 1$ or $\mathbb{E}[|y-x_2|]\ge 1$. We assume that $\mathbb{E}[|y-x_2|]\ge 1$. The maximum cost is 
    $$
    \mathbb{E}\left[\left|y-\frac{x_1+x_2}{2}\right| + \frac{|x_1-x_2|}{2}\right] = \mathbb{E}[|y-1|+1]=\mathbb{E}[|y-1|] + 1.
    $$

    Then we consider an instance $\mathbf x'=(0,x_2')$ with $x_2'=4$. The prediction $\pi=1$ has an error $\eta(\pi,\mathbf x')=\frac12$, within the error bound $\eta$.  By the strategyproofness, we must have $\mathbb{E}[y'-x_2]= \mathbb{E}[y-x_2]\ge 1$, as otherwise agent 2 gains from deviating from $x_2$ to $x_2'$. In the instance $\mathbf x'$, while the optimal maximum cost is 2,  the maximum cost induced by $y'$ is
    \begin{align*}
        \mathbb{E}\left[\left|y'-\frac{x_1+x_2'}{2}\right| + \frac{|x_1-x_2'|}{2}\right] 
        &= \mathbb{E}[|y'-2|]+2\ge 1+2=3.
    \end{align*}
    Therefore, the approximation ratio is at least $\frac32$. 
\end{proof}

When the prediction is more accurate, it is not surprising that we could expect better performance guarantees (for example, the MinMaxP mechanism is $(1+\min(1,\eta))$-approximation). 
Combining LRM and MinMaxP, Balkanski et al. \cite{balkanski2024randomized}  propose the following randomized mechanism with predictions. 

\begin{mechanism}\label{mec:line+line-classic}
    Let $q\in [0, 1]$. Given profile $\mathbf{x}$ and prediction $\pi$, run MinMaxP with probability $1-q$, and LRM with probability $q$.
\end{mechanism}

While Balkanski et al. \cite{balkanski2024randomized} show that this mechanism is $(1+\frac{q}{2})$-consistent and $(2-\frac{q}{2})$-robust, they do not consider the approximation ratio with respect to the prediction error bound. Regarding this concern, the following result presents a smooth transformation of the guarantees when the prediction error increases. The proof is a trivial combination of the guarantees of LRM and MinMaxP.

\begin{proposition}
    Mechanism \ref{mec:line+line-classic} is SP and $(1+\frac{q}{2} + (1-q)\min(1, \eta))$-approximation with respect to error bound $\eta$. 
\end{proposition}

\begin{proof}
    The strategyproofness is clear because both mechanisms, LRM and MinMaxP, are SP. For the approximation ratio with respect to error bound $\eta$, since MinMaxP is $(1+\min(1,\eta))$-approximation and LRM is $\frac{3}{2}$-approximation, we have the approximation
    \begin{align*}
        \gamma(\eta) &= \frac{3q}{2} + (1-q) (1+\min(1,\eta)) = 1\!+\!\frac{q}{2}\! +\! (1\!-\!q)\min(1, \eta).
    \end{align*}
    The consistency and robustness follow immediately by letting $\eta=0$ and $\eta\to\infty$. 
\end{proof}

When the prediction error bound is $\eta< \frac12$, clearly Mechanism \ref{mec:line+line-classic} has a better approximation ratio than the prediction-free  LRM mechanism, and when $\eta> \frac12$ it is better than the MinMaxP mechanism. If the prediction error $\eta$ is known, the decision-maker can certainly select either $q=0$ or $q=1$ to obtain the better individual mechanism. If $\eta$ is unknown, the parameter $q$ plays a role of balancing the two individual mechanisms. 

Further, the performance guarantee of Mechanism \ref{mec:line+line-classic} cannot be improved in the sense that no randomized SP mechanism guarantees $1+\delta$ consistency for some $\delta \in [0, 0.5]$ and guarantees robustness better than $2-\delta$ simultaneously \cite{balkanski2024randomized}.

\subsection{$L_p$-Norm Social Cost}\label{sec:norm}

Now, in this subsection, our focus  moves temporarily from the maximum cost objective to a more general objective of minimizing the $L_p$-norm social cost (where $p\ge 1$).  
The $L_p$-norm social cost is defined as 
\begin{align*}
    \text{SC}_p(\mathbf{x}, f(\mathbf{x}, \pi)) = \left(\sum_{i \in N} d(x_i, y)^p \right)^{\frac{1}{p}}.
\end{align*}

When $p = 1$ , it reduces to the total cost which corresponds to the notion of efficiency; and when $p \rightarrow \infty$, it reduces to the maximum cost which corresponds to fairness or equity \cite{chan2021mechanismsurvey}. Therefore, for objectives with other $p$ values $1<p<+\infty$, they can be simply and intuitively used to model the tradeoffs between the above two standard objectives (i.e., efficiency vs fairness/equity). Specifically, when $p$ increases, higher costs have a disproportionately larger influence on the objective. In other words, the social planner still considers all agents, but gives more weight to agents who experience larger costs. This captures a natural fairness adjustment: agents with large costs matter more, but not exclusively so.

Because  there exist SP optimal mechanisms for social cost without predictions \cite{procaccia2013approximate} and  the maximum cost has been discussed in earlier this section, we only need to look at the case when $1 < p < +\infty$.
Feigenbaum et al. \cite{feigenbaum2017approximately} show that the median mechanism achieves a $2^{1-1/p}$-approximation for the $L_p$-norm social cost, and this bound is tight for prediction-free mechanisms. 
When predictions are available, we continue to employ the MinMaxP mechanism (Mechanism \ref{mec:line}). 
We first introduce a useful lemma on the ``maximum" objective value, and then prove the bound of MinMaxP.

\begin{lemma}\label{lemma:norm1}
    For any profile with $x_1 \le x_2 \le \dots \le x_n$ and any finite $p \ge 1$, the maximum $L_p$-norm  social cost with $y \in [x_1, x_n]$ is attained at either $y = x_1$ or $y = x_n$.
\end{lemma}
\begin{proof}
    Note that maximizing $\text{SC}_p(\mathbf{x}, y)$ is equivalent to maximizing
    $$F(y) = \sum_{i=1}^n |x_i - y|^p,$$
    in the sense of optimal solutions. 
Let $g(z) = |z|^p$.  It is easy to see that $g$ is convex for all $p\ge 1$. 
    For each fixed $x_i$, the function $h_i(y) = g(x_i - y)$ is the composition of a convex function with an affine function $y \mapsto x_i - y$, hence $h_i(y)$ is convex. The sum $F(y) = \sum_{i=1}^n h_i(y)$ is therefore convex.
    Since a convex function on a closed interval $[x_1, x_n]$ attains its maximum at one of the endpoints, we have   $\max_{y \in [x_1, x_n]} F(y) = \max\{F(x_1), F(x_n)\}$. 
\end{proof}


\begin{theorem}
    Mechanism \ref{mec:line} is 1-consistent and $\mathcal{O}\left(n^{1/p}\right)$-robust.
\end{theorem}

\begin{proof}
    When the prediction is accurate, it must be $\pi\in [x_1, x_n]$, and the mechanism returns this prediction location, implying the 1-consistency. For the robustness, denote by $\text{ALG}$ the $L_p$-norm social cost induced by the mechanism. 
    It follows from Lemma \ref{lemma:norm1} that,  for any profile $\mathbf x$, the maximum possible $L_p$-norm social cost is achieved by the mechanism either when  $y=x_1$ or when $y=x_n$. Assume w.l.o.g. that $y=x_1$ for the worst case. We have
    \begin{align*}
        \text{ALG}^p = \sum_{i\in N}(x_i-y)^p = \sum_{i\in N}(x_i-x_1)^p \le (n-1)\cdot (x_n-x_1)^p.
    \end{align*}
 Let $y^*\in[x_1,x_n]$ be an optimal solution, and the optimal $L_p$-norm social cost $\text{OPT}$ satisfies
    \begin{align*}
        \text{OPT}^p \ge (y^*-x_1)^p+(x_n-y^*)^p\ge 2\cdot \left(\frac{x_n-x_1}{2}\right)^p,
    \end{align*}
    where the second inequality follows from the convexity of $x^p$.
    Therefore, we have
    \begin{align*}
        \frac{\text{ALG}^p}{\text{OPT}^p} \le \frac{(n-1)\cdot (x_n-x_1)^p}{2\cdot \left(\frac{x_n-x_1}{2}\right)^p} = 2^{p-1}\cdot (n-1)
    \end{align*}
    which indicates
    \begin{align*}
        \frac{\text{ALG}}{\text{OPT}} \le 2^{1-\frac{1}{p}}\cdot (n-1)^{\frac{1}{p}} = \mathcal{O}\left(n^{1/p}\right),
    \end{align*}
    establishing the robustness. 
\end{proof}

Furthermore, we show that the MinMaxP mechanism is indeed the unique  1-consistent mechanism with bounded robustness.  Before this characterization, we provide several lemmas.

\begin{lemma}\label{lemma:lp-1}
  For the $L_p$-norm social cost with $1<p<+\infty$,  if $x_1 < \pi < x_n$, then any deterministic SP $1$-consistent and bounded robust mechanism must output $y=\pi$.
\end{lemma}

\begin{proof}
 Let $f$ be a deterministic SP $1$-consistent and bounded robust mechanism.   Let $y_0=f(\mathbf x,\pi)$ be the output   for instance $\mathbf x$ with correct prediction $\pi$. Suppose for contradiction that $y_0>\pi$. By Lemma \ref{lemma:line-5}, we know $y_0\in (\pi, x_n]$. First, we move those agents with location $x_i>y_0$ to $y_0$ one by one. By Lemma \ref{lemma:sp}, the output $y_0$ does not change. Second, there must exist a location  $z<\pi$ such that if we move all agents with location $x_i<y_0$ to $z$,  $\pi$ is still an optimal solution. The existence is guaranteed because of the continuity of the social cost function and the fact that  when $1<p$ the optimal solution will never be the two extreme agent locations (as the derivative of the social cost function is non-zero at the extreme agent locations). 
  For this new instance where $n_1$ agents are at $z$ and $n-n_1$ agents are at $y_0$,  by Lemma \ref{lemma:sp}, the output $y_0$ also does not change. While the prediction $\pi$ is an optimal solution, $y_o>\pi$ must be suboptimal, giving a contradiction to the 1-consistency. 
\end{proof}

\begin{lemma}\label{lemma:lp-2}
    If $\pi\in \{x_1, x_n\}$, then any deterministic SP $1$-consistent and bounded robust mechanism must output $y=\pi$.
\end{lemma}

\begin{proof}
    We only consider the case when $\pi= x_1$. If the output is $y>x_1$, then agent 1 can decrease the cost to 0 by misreporting the location as $x_1'=x_1-\delta$ for some sufficiently small $\delta>0$ such that the output becomes $y'=\pi$ by Lemma \ref{lemma:lp-1}, leading to a contradiction to the strategyproofness. It follows from Lemma  \ref{lemma:line-5} that $y=x_1$. 
\end{proof}

\begin{lemma}\label{lemma:lp-3}
For any deterministic SP $1$-consistent and bounded robust mechanism,  if $\pi < x_1$ , then it must output $y=x_1$, and if $\pi > x_n$ it must output $y=x_n$.
\end{lemma}

\begin{proof}
    We only consider the case when $\pi < x_1$. Let $y$ be the output.  If $d(x_1, y)> x_1-\pi$, then agent 1 can decrease the cost by misreporting $x_1'=\pi$ such that the output becomes $\pi$ by Lemma~\ref{lemma:lp-2}. Hence, it must be  $d(x_1, y)\le x_1-\pi$. Since $y\in [x_1, x_n]$ by the 1-consistency,  we have $y\in [x_1, \min(x_n, 2x_1-\pi)]$.  If $y>x_1$, let agent 1 move to $x_1'=\pi$, and by Lemma \ref{lemma:sp} the output for the new instance remains to be $y$. However, Lemma \ref{lemma:lp-2} claims that the output should be the prediction $\pi=x_1'$, giving a contradiction.  Therefore, the mechanism must output $y=x_1$ when $\pi<x_1$.
\end{proof}

Combining Lemmas~\ref{lemma:lp-1}, \ref{lemma:lp-2} and \ref{lemma:lp-3} gives the following characterization.

\begin{theorem}
   For the $L_p$-norm social cost with $1<p<+\infty$, MinMaxP is the unique deterministic SP mechanism with $1$-consistency and bounded robustness.
\end{theorem}

\section{The Tree Metric Space}\label{sec:tree}

In this section, we consider the problem on a tree $T$. 
A tree $T \subset \mathbb{R}^m$ is a closed and connected subset of the $m$-dimensional Euclidean space that consists of a finite number of closed curves of finite length, referred to as the \emph{edges}. 
The endpoints of these curves are called \emph{vertices}. 
Both the agents and the facility may be located anywhere on $T$. 
The distance between two points $x, y \in T$, denoted by $d(x, y)$, is defined as the length of the shortest path connecting $x$ and $y$, 
where a \emph{path} is the unique minimal connected subset of $T$ containing both $x$ and $y$. 
The \emph{center} of the path between $x$ and $y$, denoted by $\text{cen}(x, y)$, is the point $z$ on the path such that $d(x, z) = d(y, z)$. 
Because of the structural property of trees, there are no \emph{cycles}, where a cycle is defined as the union of two paths whose intersection consists exactly of their two endpoints.

Given the profile $\mathbf x$ of agents on a tree $T$,  the optimal facility location that minimizes the maximum cost is the center of the longest path between any two agents, denoted by  $\text{cen}(T,\mathbf x)$.
For deterministic strategyproof mechanisms without prediction, Alon et al. \cite{alon2010strategyproof} prove the tight bound of 2 for minimizing the maximum cost. Relying on the prediction, we improve this bound for deterministic and randomized mechanisms in Sections~\ref{sec:tree-d} and \ref{sec:tree-r}, respectively.

\subsection{Deterministic Mechanisms}\label{sec:tree-d}

We first define the normalization of a point with respect to a given profile.
We use $\text{normal}(y, T, \mathbf{x})$ to denote the normalization of a point $y$ for agents located on the tree~$T$. 
If there are agents on at least two sides of $y$, or if $y$ coincides with some agent, then $\text{normal}(y, T, \mathbf{x})$ simply returns $y$; 
otherwise, it moves in the direction where the agents are located until reaching either an agent or a point where agents exist on at least two sides. 
When the tree is a line, we have $\text{normal}(y, T, \mathbf{x}) = \max(x_1, \min(y, x_n)).$

If a point $y$ satisfies $\text{normal}(y, T, \mathbf{x}) \ne y$, we say that $y$ is \emph{outside} the agents. 
If $\text{normal}(y, T, \mathbf{x}) = y$, we say that $y$ is \emph{inside} the agents.
Analogous to the MinMaxP mechanism on the line, we define the Tree MinMaxP mechanism as follows, and extend the performance guarantees by Agrawal et al. \cite{agrawal2022learning} from line to tree.


\begin{mechanism}[Tree MinMaxP]\label{mec:tree}
     Given profile $\mathbf{x}$ and prediction $\pi$ on a tree $T$, return $\text{normal}(\pi, T, \mathbf{x})$.
\end{mechanism}

\begin{theorem}
    Mechanism \ref{mec:tree} (Tree MinMaxP) is SP and $(1+\min(1, \eta))$-approximation. In addition, it is 1-consistent and 2-robust.
\end{theorem}

\begin{proof}
    First, we prove the strategyproofness. If $\pi$ is inside the agents, the output is $\pi$ itself, and the only possibility to change the output is when all agents on one side misreport their locations to another side; however, this would only increase their distance to the output. If $\pi$ is outside the agents, the output is the closest point from $\pi$ that lies in some path of two agents. To change the output,  there must exist a direction of this output such that all agents on this side misreport to some other side, and this would only increase cost. 
    

    The $1$-consistency follows from the fact that a correct prediction $\pi$ must be inside the agents, as otherwise the point $y=\text{normal}(\pi,T,\mathbf{x})$ will decrease all the agents' cost, contradicting to the correctness of $\pi$. 
The 2-robustness follows from, for any $i, j\in N$,
    $$MC(\mathbf x,f(\mathbf x,\pi))=\max_{i\in N}d(x_i, y)\le \max_{i,j\in N}d(i, j)=2\cdot\text{OPT},$$
    where the inequality is because the output $y$ is always inside the agents, and the second equality is a simple observation. 
    When $\eta \le 1$, denoting by $o$ the optimal solution, we have  $$MC(\mathbf x,f(\mathbf x,\pi))=\max_{i\in N}d(x_i, y)\le \max_{i\in N}d(x_i,o)+d(o, y)=\text{OPT}+\eta\cdot \text{OPT}=(1+\eta)\cdot \text{OPT}.$$
   Combined with the 2-robustness, it establishes a $(1+\min(1, \eta))$-approximation.
\end{proof}

The matching lower bound $1+\min(1, \eta)$ naturally follows from 
Theorem \ref{thm:lowl} because the line is a special tree.

\subsection{Randomized Mechanisms}\label{sec:tree-r}

Alon et al. \cite{alon2010strategyproof}  propose the Tree Random Mechanism (without prediction): return each $x_i$ with probability $\frac{1}{n+2}$ and $\text{cen}(T,\mathbf{x})$ with probability $\frac{2}{n+2}$. They show the strategyproofness and an approximation ratio of $\frac{2n+2}{n+2}$. Combining with the Tree MinMaxP mechanism, we propose the following randomized one.



\begin{mechanism}\label{mec:tree+tree-ran} Let $q\in [0, 1]$. Given profile $\mathbf{x}$ and prediction $\pi$ on a tree $T$, run the Tree MinMaxP (Mechanism \ref{mec:tree}) with probability $1-q$ and the Tree Random Mechanism \cite{alon2010strategyproof}  with probability $q$.
\end{mechanism}

\begin{theorem}
    Mechanism \ref{mec:tree+tree-ran} is SP and $(1+\frac{n}{n+2}q+(1-q)\min(1,\eta))$-approximation. Additionally, it is $(1+\frac{n}{n+2}q)$-consistent and $(2-\frac{2}{n+2}q)$-robust.
\end{theorem}

\begin{proof}
    Since both individual mechanisms are SP, Mechanism \ref{mec:tree+tree-ran}  is also SP.
    For the consistency, since Tree MinMaxP is 1-consistent and the Tree Random Mechanism is $\frac{2n+2}{n+2}$-consistent, we have the consistency ratio of
    $$
 q\cdot \frac{2n+2}{n+2} + (1-q)\cdot 1 = 1+\frac{n}{n+2}q.
    $$
For the robustness, since Tree MinMaxP is $2$-robust and the Tree Random Mechanism  is $\frac{2n+2}{n+2}$-robust, we have the consistency ratio of
    $$
     q\cdot \frac{2n+2}{n+2} + (1-q)\cdot 2=2-\frac{2}{n+2}q.
    $$
Moreover, by the approximation ratio  $(1+\min(1,\eta))$ of Tree MinMaxP, Mechanism \ref{mec:tree+tree-ran}  has an approximation ratio of
    $$
     q \cdot \frac{2n+2}{n+2} + (1-q)\cdot (1+\min(1, \eta)) = 1+\frac{n}{n+2}q+(1-q)\min(1,\eta).
    $$
\end{proof}

When the prediction error bound is $\eta< \frac{n}{n+2}$, clearly Mechanism \ref{mec:tree+tree-ran} has a smaller approximation ratio than the prediction-free Tree Random Mechanism, and when $\eta> \frac{n}{n+2}$ it is better than the Tree MinMaxP mechanism. The selection of the value of $q$ could balance the two individual mechanisms when $\eta$ is unknown. 

\section{Two-dimensional Metric Spaces}\label{sec:2d}

In this section, we study the problem on the 2-dimensional space $(\mathbb R^2,d_p)$ with $L_p$ metrics, where the distance of two points $x,y\in\mathbb R^2$ is 
$$d_p(x, y) =(\sum_{i=1}^{2} |x_i - y_i|^p)^{1/p}.$$
We denote the space by $\ell_2^p$.  

For $\ell_2^1$ space or even $\ell_n^1$ spaces, since the distance is coordinate-wise independent, we can run the mechanisms in Section \ref{sec:line} for each coordinate, and retain strategyproofness and all performance guarantees as in the line.
Hence, we only consider $\ell_2^p$ spaces with $p\ge 2$. 
For convenience, we denote the location of agent $i$ as $x_i=(a_i, b_i)$, the prediction as $\pi=(a_{\pi}, b_{\pi})$ and the output as $y=(a_y, b_y)$. Denote $a_{\min}=\min_{i\in N}a_i$ and $a_{\max}=\max_{i\in N}a_i$, and these notations are similar for coordinate $b$. 
We study deterministic mechanisms in Section \ref{subsec:der2} and randomized mechanisms in Section \ref{subsec:ran2}.

\subsection{Deterministic Mechanisms}\label{subsec:der2}

Consider the Minimum Bounding Box mechanism proposed in  \cite{agrawal2022learning}. If the minimum rectangle that contains all
the agent locations, contains the prediction point $\pi$, then we output $\pi$; otherwise we select the boundary point with the minimum distance from $\pi$. Equivalently, it runs MinMaxP  for each dimension separately. 
Agrawal et al. \cite{agrawal2022learning} prove that it achieves a $1+\min\{\eta,\sqrt 2\}$ approximation in the $\ell_2^2$ space, implying the 1-consistency and $(1+\sqrt2)$-robustness. We extend their results to $\ell_2^p$ spaces for any $p\ge 2$. 

\begin{mechanism}[Minimum Bounding Box]\label{mec:2d}
 Given location profile $\mathbf{x}$ and prediction $\pi$,  return $y=(a_y, b_y)$ with $a_y = \max(a_{\min}, \min(a_{\max}, a_{\pi}))$ and $b_y = \max(b_{\min}, \min(b_{\max}, b_{\pi}))$.
\end{mechanism}

\begin{theorem}
     Minimum Bounding Box is SP and has approximation ratio of $1+\min\{\eta,2^{\frac{1}{p}}\}$ with respect to prediction error bound $\eta$. In addition, it is $1$-consistent and $(1+2^{\frac{1}{p}})$-robust.
\end{theorem}

\begin{proof}
For the strategyproofness, notice that the mechanism handles each dimension independently by applying the MinMaxP mechanism to each one. Therefore, since agents prefer the facility to be as close as possible to their own coordinate in every dimension, the strategyproofness of MinMaxP ensures that the Minimum Bounding Box mechanism is also strategyproof.

When $\eta\ge 2^{1/p}$, we prove it to be $(1+2^{\frac{1}{p}})$-approximation. We draw the circle with the smallest radius that contains all agents inside. The radius $r$ is the optimal maximum cost and the center $o=(a_o, b_o)$ is the optimal facility location. Therefore, the maximum cost of this mechanism is
    \begin{align*}
        \text{MC}(\mathbf x,y)&=\max_{i\in N}d(x_i, y)\le \max_{i\in N}d(x_i, o)+d(o, y)\le r+\left(\max_{i\in N}|a_i-a_o|^p+\max_{i\in N}|b_i-b_o|^p\right)^{\frac{1}{p}}\\
        &\le r + \left(2r^p\right)^{\frac{1}{p}}=\left(1+2^{\frac{1}{p}}\right) r,
    \end{align*}
    where the second inequality is because $a_y$ lies in the interval $[a_{\min},a_{\max}]$ and $b_y$ lies in $[b_{\min},b_{\max}]$.

When $\eta<2^{1/p}$, we prove it to be $(1+\eta)$-approximation. Again, we draw the circle with the smallest radius that contains all agents. The radius $r$ is the optimal maximum cost and the center $o=(a_o, b_o)$ is the optimal facility location. By the definition of the prediction error $\eta$, we have $|a_\pi-a_o|^p+|b_\pi-b_o|^p\le (\eta \cdot r)^p$, and
    \begin{align*}
        \text{MC}(\mathbf x,y)&=\max_{i\in N}d(x_i, y)\le \max_{i\in N}d(x_i, o)+d(o, y)\le r + \left(\min(|a_\pi-a_o|, r)^p+\min(|b_\pi-b_o|, r)^p\right)^{\frac{1}{p}}\\
        &\le r + \left(|a_\pi-a_o|^p+|b_\pi-b_o|^p\right)^{\frac{1}{p}}\le (1+\eta)\cdot r.
    \end{align*}
Therefore, the approximation ratio is at most $1+2^{\frac{1}{p}}$ when $\eta\ge 2^{1/p}$, and  $1+\eta$ when $\eta< 2^{1/p}$. The consistency and robustness follow immediately by letting $\eta=0$ and $\eta\to\infty$. 
\end{proof}

The approximation analysis of the mechanism is tight. Consider the following example. In a unit circle $a^p+b^p=1$  there are three agents with $x_1=(-\frac{1}{2^{1/p}}, -\frac{1}{2^{1/p}})$, $x_2=(0, 1)$, $x_3=(1,0)$. The optimal facility location of this instance is $(0, 0)$ and the optimal maximum cost is 1. For any $\eta\ge 0$, let $\pi=\left((\frac{\eta^p}{2})^\frac{1}{p}, (\frac{\eta^p}{2})^\frac{1}{p}\right)$ be the prediction, and the error is equal to $\eta$. When $\eta\le 2^\frac{1}{p}$, the output of the mechanism is $\pi$ and the maximum cost is $1+\eta$ (attained by agent 1), giving an approximation ratio of at least $1+\eta$. When $\eta> 2^\frac{1}{p}$, the output is $(1, 1)$, and the maximum cost is $1+2^{\frac{1}{p}}$ (attained by agent 1), giving an approximation ratio of at least $1+2^{\frac{1}{p}}$.

While it is shown in \cite{agrawal2022learning} that for the $\ell_2^2$ space there is no deterministic  SP mechanism that is $(2-\epsilon)$-consistent and $(1+2^{1/p}-\epsilon)$-robust for any $\epsilon>0$, it is challenging to generalize this impossibility result, and we conjecture it also holds for general $\ell_2^p$ spaces with $p\ge 2$. 

\subsection{Randomized Mechanisms}\label{subsec:ran2}

Before studying randomized mechanisms, we start with a deterministic prediction-free mechanism, the \emph{coordinate-wise median} (CM) mechanism. CM is a classic and most well-studied strategy-proof mechanism in social choice and facility location problems \cite{feigenbaum2017approximately,DBLP:journals/scw/GoelH23}, in which the facility's coordinates are separately computed as medians of agent coordinates in each dimension. In our setting,  given profile $\mathbf{x}$, CM always returns $a_y$ as the median of all $a_i$ ($1\le i\le n$) and $b_y$ as the median of all $b_i$. We first show that it has $2$-approximate the maximum cost, and then propose a novel randomized mechanism with predictions that combines CM and Minimum Bounding Box. 

\begin{proposition}\label{prop:cm}
The CM mechanism is SP and $2$-approximate for the maximum cost. 
\end{proposition}

\begin{proof}
  Firstly we prove it is SP with the same method of Minimum Bounding Box (Mechanism \ref{mec:2d}). If an agent $i$ wants $a_{y}$ to be larger, then it has to misreport from some $a_{i}\le a_{y}$ to some $a_{i}' > a_{y}$, which will increase the cost. In the same way, if the agent wants $a_{y}$ to be smaller, or wants $b_{y}$ to be smaller or larger, then the cost will increase. This means changing output in one dimension cannot decrease the cost. Changing output in both of the two dimensions will also increase the cost because in both dimensions, the distance will increase, and the total distance is calculated by $(|a_{i}'-a_{y}'|^p+|b_{i}'-b_{y}'|)^{\frac{1}{p}}$.

    Then we prove the approximation ratio. We still draw the circle with the smallest radius that contains all agents. The radius $r$ is the optimal maximum cost and the center $o=(a_o, b_o)$ is the optimal output. We claim that $y$ will never be outside the circle. If that claim holds, then since all the agents and output are not outside the circle, the maximum distance between an agent and output is the length of diameter $2r$, leading to 2-approximation. Next we prove this claim by contradiction.

    Suppose some output $y$ satisfies $(a_y-a_o)^p+(b_y-b_o)^p>r^p$. Notice that $a_y, b_y$ are the median of all $a_i$ and $b_i$, we know $a_y\in [a_o-r, a_o+r]$ and $b_y\in [b_o-r, b_o+r]$. Notice that if $a_y=a_o$, then $|b_y-b_o|>r$ which is contradicted with $b_y\in [b_o-r, b_o+r]$. Therefore $a_y\ne a_o$ and $b_y\ne b_o$. Due to symmetry, we can assume $a_y>a_o$ and $b_y>b_o$.

    Suppose the two coordinate axes are $a$ and $b$ axes. Draw line $b=b_y$ intersecting the circle at $P_1, P_2$ ($P_1$ is on the left), and at least half agents satisfy $b_i\ge b_y$. Draw line $a=a_y$ intersecting the circle at $Q_1, Q_2$ ($Q_1$ is at the top), and at least half agents satisfy $a_i\ge a_y$. Notice there is on overlap for the two regions above, so all the agents are in these two regions. Connect $P_1$ and $Q_2$, we can calculate that $o$ is under this line. We prove this in the end of the proof. If that holds, then draw a diameter of a circle through point $o$ that is parallel to $P_1 Q_2$. Of the two semicircles divided by this diameter, the lower semicircle (including its boundary) does not contain any agent. Therefore this is not the smallest circle, leading into a contradiction.

    Last we prove that $o$ is under this line. For convenience, we translate all the points and graphs to let $o=(0,0)$. In this way the equation of the circle is $a_y^p+b_y^p=r^p$. By using equations, we can calculate that $P_1=(-(r^p-b_y^p)^{\frac{1}{p}}, b_y)$, and $Q_2=(a_y, -(r^p-a_y^p)^{\frac{1}{p}})$. Therefore the line equation is
    $$
    b = -\frac{b_y + (r^p-a_y^p)^{\frac{1}{p}}}{a_y + (r^p-b_y^p)^{\frac{1}{p}}}(a - a_y) - (r^p-a_y^p)^{\frac{1}{p}}.
    $$
    When $a=0$, we have
    \begin{align*}
        b &= a_y\cdot \frac{b_y + (r^p-a_y^p)^{\frac{1}{p}}}{a_y + (r^p-b_y^p)^{\frac{1}{p}}} - (r^p-a_y^p)^{\frac{1}{p}}\\
        &= \frac{a_yb_y + a_y(r^p-a_y^p)^{\frac{1}{p}} - a_y(r^p-a_y^p)^{\frac{1}{p}} - (r^p-a_y^p)^{\frac{1}{p}}(r^p-b_y^p)^{\frac{1}{p}}}{a_y + (r^p-b_y^p)^{\frac{1}{p}}}\\
        &= \frac{a_yb_y-(r^p-a_y^p)^{\frac{1}{p}}(r^p-b_y^p)^{\frac{1}{p}}}{a_y + (r^p-b_y^p)^{\frac{1}{p}}}
    \end{align*}
    If we can prove $b > 0$, then $o$ is under this line. Therefore we only need to prove
    \begin{align*}
        a_yb_y > (r^p-a_y^p)^{\frac{1}{p}}(r^p-b_y^p)^{\frac{1}{p}}\Longleftrightarrow \quad & a^pb^p > (r^p-a_y^p)(r^p-b_y^p)\\
        \Longleftrightarrow \quad & r^p\cdot (a_y^p+b_y^p) > r^{2p}\\
        \Longleftrightarrow \quad & a_y^p+b_y^p > r^p.
    \end{align*}
    Therefore we complete the proof.
\end{proof}

\begin{mechanism}
\label{mec:2d+gcm}
     Let $q\in [0, 1]$. Given profile $\mathbf{x}$ and prediction $\pi$, run the Minimum Bounding Box mechanism with probability $1-q$, and CM with probability $q$.
\end{mechanism}

\begin{theorem}
    Mechanism \ref{mec:2d+gcm} is SP and $(1+q + (1-q)\min(2^{\frac{1}{p}}, \eta))$-approximation with respect to prediction error bound $\eta$. In addition, it is $(1+q)$-consistent and $(1+q+2^{\frac{1}{p}}(1-q))$-robust.
\end{theorem}

\begin{proof}
    Since Minimum Bounding Box and CM are both SP, we know Mechanism \ref{mec:2d+gcm} is also SP. 
    For the approximation ratio, since Minimum Bounding Box is $(1+\min(2^{\frac{1}{p}},\eta))$-approximation and CM is $2$-approximation, we have the approximation ratio
    \begin{align*}
        \gamma(\eta) &= q \cdot 2 + (1-q)\cdot (1+\min(2^{\frac{1}{p}},\eta))= 1+q + (1-q)\min(2^{\frac{1}{p}}, \eta).
    \end{align*}
The consistency and robustness follow immediately by letting $\eta=0$ and $\eta\to\infty$. 
\end{proof}

We notice that when the prediction error bound is $\eta< 1$,  Mechanism \ref{mec:2d+gcm} has a better approximation ratio than the prediction-free mechanism CM, and when $\eta> 1$, it is better than the Minimum Bounding Box mechanism.

\section{Group Strategyproofness}\label{sec:ext}



A mechanism is \emph{group strategyproof} (GSP) if no group of agents can collude to misreport their preferences in a way that makes every member better off. That is,  a mechanism is GSP if for all $\mathbf x$, all $S\subseteq N$, and all $\mathbf x_S'\in\mathcal M^{|S|}$, there exists $i\in S$ such that $c(x_i, f(\mathbf{x}_S', \mathbf{x}_{-S}, \pi)) \ge c(x_i, f(\mathbf{x}, \pi))$. 
A mechanism is \emph{strong group strategyproof} (SGSP) if no group can collude to misreport their preferences in a way that makes at least one member better off without making any of the remaining members worse off.
It is clear that SGSP implies GSP, and GSP implies SP. All these notions are well studied in the social choice literature \cite{procaccia2013approximate,tang2020characterization}. We examine whether the mechanisms considered are GSP or SGSP and complete the picture by lower bound results. 

\paragraph{The real line.}    Mechanism \ref{mec:line} (MinMaxP) is SGSP. When $x_1\le \pi\le x_n$, the output is $\pi$, and to change the outcome in the new instance $\mathbf x'$ it must be $x_1'>\pi$ or $x_n'<\pi$. In the former case, agent 1 must be a group member but she increases the cost. In the latter case, agent $n$ must be a group member but she also increases the cost. 
When $x_1>\pi$, the output is $x_1$. Agent 1 will never join the group, and the only possibility to change to output is to move it to the left; however, this would increases the cost of all agents. 

Mechanism \ref{mec:line-classic} (LRM) is known to be GSP \cite{procaccia2013approximate}. However, it is not SGSP. Consider the location profile $\mathbf{x}=(0, 1, 2)$. If all the agents report $1$ as their locations, then the output will be 1, which decreases agent 2's cost to 0 and does not change the cost of agents 1 and 3. 
 Mechanism \ref{mec:line+line-classic} is  GSP but not SGSP, because it is a combination of MinMaxP and LRM. 

Next, we present a lower bound 2 on the approximation ratio of all randomized SGSP mechanisms. It indicates that the 2-robust MinMaxP achieves the best possible robustness guarantees among all SGSP mechanisms. 

\begin{lemma}\label{lemma:line-sgsp-1}
    For any randomized prediction-free SGSP mechanism with an approximation ratio at most 2, given location profile $\mathbf{x}=(x_1, x_2, x_3)$ on the line, then either $\mathbb{P}[f(\mathbf{x}) \in [x_1, x_2]]=1$ or $\mathbb{P}[f(\mathbf{x}) \in [x_2, x_3]]=1$.
\end{lemma}

\begin{proof}
    It is easy to see that the 2-approximation guarantees that $\mathbb{E}[f(\mathbf{x})]\in[x_1,x_3]$. 
    W.l.o.g., we assume $\mathbb{E}[f(\mathbf{x})] = c_0 \in [x_1, x_2]$. Consider another profile $\mathbf{x}'=(c_0, c_0, c_0)$ and obviously $\mathbb P[f(\mathbf{x}')=c_0]=1$. If agents misreport from $\mathbf{x}$ to $\mathbf{x}'$, then no agent will increase the cost because for any $i\in N$, 
    $$
    c(x_i, f(\mathbf{x})) = \mathbb{E}[|x_i- f(\mathbf{x})|]\ge |x_i-\mathbb{E}[f(\mathbf{x})]| = c(x_i, f(\mathbf{x}')).
    $$
    The SGSP property further ensures that no agent could decrease the cost in the new instance, that is,  $\mathbb{E}[|x_i- f(\mathbf{x})|] = |x_i-\mathbb{E}[f(\mathbf{x})]|$ for any $i\in N$. It implies that all realizations of $f(\mathbf{x})$ should be at the same side of $x_i$. More formally, with probability 1,  $f(\mathbf{x})$ should belong to one of the intervals $(-\infty, x_1]$, $[x_1, x_2]$, $[x_2, x_3]$ and $[x_3, \infty)$. Recalling that $\mathbb{E}[f(\mathbf{x})] \in [x_1, x_3]$,  $f(\mathbf{x})$ must belong to one of the intervals $[x_1, x_2]$ and $[x_2, x_3]$ with probability 1.
\end{proof}

\begin{proposition}\label{prop:ext}
    No randomized prediction-free SGSP mechanism on a line can be better than 2-approximation. It implies that no randomized SGSP mechanism with predictions can be better than 2-robust.
\end{proposition}

\begin{proof}
    Suppose $f$ is such a mechanism. Consider $\mathbf{x}=(0,1,2)$. By Lemma \ref{lemma:line-sgsp-1}, assume w.l.o.g. that $\mathbb{P}[f(\mathbf{x}) \in [x_1, x_2]]=1$. Then we move $x_3$ towards positive direction with 0.01 in each step. Denote by $x_3^{(k)}=2+0.01k$ the location of agent 3 after the $k$-th step. Denote by $\mathbf{x}^{(k)}$ the location profile. 

    We prove by induction that, after each step $k$,  $f(\mathbf{x}^{(k)})\in[x_1, x_2]$ with probability 1 and $\mathbb{E}[f(\mathbf{x}^{(k)})] = \mathbb{E}[f(\mathbf{x}^{(0)})]$. It is clearly true for $k=0$. 
    For $k=1$,     to prevent agent 3 misreporting from $x_3^{(0)}$ to $x_3^{(1)}$, the strategyproofness ensures that $\mathbb{E}[d(x_3^{(0)}, f(\mathbf{x}^{(1)}))]\ge \mathbb E[d(x_3^{(0)}, f(\mathbf{x}^{(0)}))]$. If $f(\mathbf{x}^{(1)})$ is in $[x_2, x_3^{(1)}]$, we have
    \begin{align*}
        d(x_2, x_3^{(0)}) &= \max_{y\in [x_2, x_3^{(1)}]} d(y, x_3^{(0)}) \ge \mathbb{E}[d(x_3^{(0)}, f(\mathbf{x}_3^{(1)}))]\ge \mathbb E[d(x_3^{(0)}, f(\mathbf{x}_3^{(0)}))]\ge \min_{y\in [x_1, x_2]}d(y, x_3^{(0)})\\
        &= d(x_2, x_3^{(0)}).
    \end{align*}
    This indicates that all the inequalities should be equalities, and that  $f(\mathbf{x}^{(1)})$ and $f(\mathbf{x}^{(0)})$ are both at $x_2$.
    Similarly,   if $f(\mathbf{x}^{(1)})$ is in $[x_1, x_2]$, we have
    \begin{align*}
        d(\mathbb{E}[f(\mathbf{x}^{(1)})], x_3^{(0)}) &= \mathbb{E}[d(x_3^{(0)}, f(\mathbf{x}^{(1)}))]\ge \mathbb E[d(x_3^{(0)}, f(\mathbf{x}^{(0)}))]= d(\mathbb{E}[f(\mathbf{x}^{(0)})], x_3^{(0)}).
    \end{align*}
    This indicates that $d(\mathbb{E}[f(\mathbf{x}^{(1)})], x_3^{(0)})=d(\mathbb{E}[f(\mathbf{x}^{(0)})], x_3^{(0)})$. Therefore, combining these two cases, we have  $\mathbb P[f(\mathbf{x}^{(1)})\in[x_1, x_2]]=1$ and $\mathbb{E}[f(\mathbf{x}^{(1)})] = \mathbb{E}[f(\mathbf{x}^{(0)})]$. Applying the same arguments for the $k$-th step, we complete the induction proof.
    
    When $k\rightarrow \infty$, $x_3^{(k)}\rightarrow \infty$, and the approximation ratio of the mechanis approaches 2.
\end{proof}

\paragraph{The tree.} Mechanism~\ref{mec:tree} (Tree MinMaxP) is SGSP. 
If the output moves toward some direction, then all the agents located on the opposite side, as well as the output itself, would need to misreport their locations toward that direction. 
However, doing so would only increase the cost for all of them.
In contrast, the Tree Random Mechanism of \cite{alon2010strategyproof} is not GSP. 
Consider a star with four vertices $O, A, B, C$ and three unit-length edges $OA, OB, OC$. 
There are three agents located at $A$, $B$, and $C$, respectively. 
Under this mechanism, the costs for agents $A$, $B$, and $C$ are all equal to
$\frac{2}{5}\cdot 2 + \frac{2}{5}\cdot 1 = \frac{6}{5}$.
If all agents misreport their locations to $O$, the output would be $O$ deterministically and their cost decreases to $1$. 
It follows that Mechanism~\ref{mec:tree+tree-ran} is also not GSP.

\paragraph{$\ell_2^p$ metric spaces.}   When $p=1$, running the mechanisms on the line for each dimension maintains the properties. When $p\ge 2$, all mechanisms considered in Section \ref{sec:2d}
are not GSP. Indeed, all deterministic GSP mechanisms must be dictatorial (that always return the location of a dictator agent), and all randomized GSP mechanisms must be 2-dictatorial, as explained in \cite{tang2020characterization}. It implies that no (randomized) prediction-free GSP  mechanism has an approximation ratio better than 2. Since this characterization also works for the setting with predictions, with  any prediction error $\eta\ge 0$, no (randomized) GSP  mechanism with prediction has approximation ratio less than 2.

\section{Conclusion}


In this work, we studied facility location mechanisms within the learning-augmented framework on the real line, on trees, and in two-dimensional \(\ell_2^p\) metric spaces. Focusing on SP mechanisms that approximately minimize the maximum cost, we parameterized their approximation guarantees in terms of the prediction error and quantified their consistency and robustness. Our results show that learning-augmented facility location mechanisms can be effectively extended beyond one-dimensional domains.

However, the landscape for \(\ell_2^p\) spaces remains only partially understood. For instance, while it is known that when \(p = 2\) no deterministic SP mechanism can be simultaneously \((2-\epsilon)\)-consistent and \((1+2^{1/p}-\epsilon)\)-robust \cite{agrawal2022learning}, it is still unclear whether this limitation extends to general \(\ell_2^p\) spaces. It would also be interesting to investigate learning-augmented mechanisms in other metric spaces, such as graphs or higher-dimensional Euclidean (or \(l^p\)) spaces.

Furthermore, the problem of multi-facility location remains unsolved. 
The conference version of this paper~\cite{ecai2025} discusses this issue and provides a lower bound: no deterministic SGSP mechanism can achieve $(2-\epsilon)$-consistency (for any $\epsilon>0$) while maintaining bounded robustness. 
For SP/GSP mechanisms, the same result holds in the following two deterministic cases: $2$-facility with anonymous mechanisms, and $m$-facility with $m>2$. 
The former is a corollary of~\cite{fotakis2013strategyproof}, where it is proved that for $2$-facility locations without predictions, any deterministic anonymous SP mechanism with a bounded approximation ratio must be the ``two-extreme'' mechanism, which always selects the leftmost and rightmost agents. 
Therefore, if bounded robustness is desired, the two-extreme mechanism is the only possibility, and the prediction can never improve; if one aims to improve consistency, the robustness becomes unbounded. 
The latter case stems from the fact that deterministic strategyproof mechanisms do not admit any bounded approximation ratio for social cost or maximum cost.

Beyond the theoretical settings considered, several important practical dimensions remain open for future work. As noted in the introduction, predictions in real systems are often derived from historical data, which may introduce dynamics not captured by our static model. For instance, if agents are aware that their reports can influence future predictions, they may adopt time-dependent strategic behaviors. Incorporating such dynamic strategic considerations, alongside other practical constraints like partial participation, noisy location reports, and budget limitations, would bring the model closer to real-world deployment scenarios. Investigating the design and performance of learning-augmented mechanisms under these more complex yet realistic assumptions is a compelling direction for further research.

\paragraph{Acknowledgements.} Hau Chan is supported by the National Institute of General Medical Sciences of the National Institutes of Health [P20GM130461], the Rural Drug Addiction Research Center at the University of Nebraska-Lincoln, and the National Science Foundation under grants IIS:RI \#2302999 and IIS:RI \#2414554. The work is supported in part by the Guangdong Provincial/Zhuhai Key Laboratory of IRADS (2022B1212010006) and by Artificial Intelligence and Data Science Research Hub, BNBU, No. 2020KSYS007. Chenhao Wang is supported by NSFC under grant 12201049, and by BNBU under grant UICR0400004-24B. The content is solely the responsibility of the authors and does not necessarily represent the official views of the funding agencies.

\bibliographystyle{plain}
\bibliography{mybibfile}

\end{document}